\documentclass[11pt,a4paper]{article}
\usepackage{jcappub}

\usepackage{amsmath}
\usepackage{amsthm}
\usepackage{graphicx}
\usepackage{amsfonts}

\newtheorem{lemat}{Lemma}

\title{Slow-roll approximation in loop quantum cosmology} 

\author{Joanna Luc and}
\author{Jakub Mielczarek}

\affiliation{Institute of Physics, Jagiellonian University, {\L}ojasiewicza 11, 30-348 Cracow, Poland}

\emailAdd{jakub.mielczarek@uj.edu.pl}

\abstract{The slow-roll approximation is an analytical approach to study
dynamical properties of the inflationary universe. In this article, systematic 
construction of the slow-roll expansion for effective loop quantum 
cosmology is presented. The analysis is performed up to the fourth order 
in both slow-roll parameters and the parameter controlling the strength 
of deviation from the classical case. The expansion is performed for 
three types of the slow-roll parameters: Hubble slow-roll parameters, 
Hubble flow parameters and potential slow-roll parameters. An accuracy 
of the approximation is verified by comparison with the numerical 
phase space trajectories for the case with a massive potential term. 
The results obtained in this article may be helpful in the search for 
the subtle quantum gravitational effects with use of the cosmological 
data.} 


\begin{document}

\maketitle

\bibliographystyle{abbrv}

\section{Introduction}

The inflationary epoch provides a unique window to test various generalizations
of the well-grounded physical theories. This, in particular, concerns supersymmetric 
extensions of the Standard Model of elementary interactions \cite{Randall:1995dj,
HenryTye:2006uv,Allahverdi:2006iq,Baumann:2014nda}. The inflationary phase may also serve as 
a testing ground for quantum extensions of General Relativity. Awareness of the 
second possibility has significantly increased in the recent years, which is reflected by
the number of studies performed within various approaches to quantum gravity. The 
attempts to confront cosmological data, throughout inflation, with such approaches to 
quantum gravity as Loop Quantum Gravity \cite{Mielczarek:2009zw,Barrau:2013ula,
Agullo:2013ai}, Causal Dynamical Triangulations \cite{Mielczarek:2015cja,
Mielczarek:2015xyu}, Horava-Lifschitz gravity \cite{Mukohyama:2010xz}, Canonical 
Quantum Gravity \cite{Kiefer:2011cc} have been made.    

The simplest model of inflationary period is given by the scalar field theory undergoing 
the slow-roll type of evolution. In the light of the WMAP data \cite{Komatsu:2010fb} 
the inflationary model with the quadratic potential (convex potential function) served 
as the most conservative explanation of the available observations. The model is, 
however, not favored by the combined up to date Planck and BICEP II observations 
\cite{Ade:2015tva}. Instead, models with either non-minimal coupling between scalar 
matter and gravity (Starobinsky model \cite{Starobinsky:1980te}) or upside-down 
(concave potential function) gain observational support. In both cases, the slow-roll 
conditions are satisfied.
 
Due to increasing precision of the observational data, reliable discrimination between 
the two models mentioned above as well as other possibilities requires departure 
beyond the linear order of the slow-roll approximation. Because of this, systematic 
analysis of the more subtle, higher order, contributions gains justification. This is 
especially important in the situation when deviations from the canonical inflationary 
models are expected to be very small. The quantum gravitational corrections to 
the inflationary dynamics are often expected to be of this kind. 

In this article, we study a phase of inflation with a particular type of the quantum 
gravitational corrections, namely with the so-called \emph{holonomy corrections}
\cite{Bojowald:2006da,Ashtekar:2011ni}. The term ``corrections'' is usually associated 
with a small deviation. In case of the holonomy corrections this is not true in the deep 
Planck epoch, where the holonomy corrections lead to non-perturbative modifications 
of dynamics generating such effects as: singularity avoidance \cite{Ashtekar:2006rx}, 
signature change \cite{Cailleteau:2011kr,Bojowald:2011aa,Bojowald:2015gra} and 
the state of silence \cite{Mielczarek:2012tn,Mielczarek:2014kea}. On the other hand, 
during inflation the holonomy corrections can be treated as small deviations from the 
standard case. 

While the inflation with holonomy corrections has already been a subject of quite broad 
studies, the linear slow-roll approximation was mostly discussed in this context 
\cite{Singh:2006im,Artymowski:2008sc,Ashtekar:2009mm,Mielczarek:2010bh, 
Mielczarek:2013xaa}. Furthermore, the definitions of the holonomy-corrected 
lowest order potential slow-roll parameters have been introduced on the basis of 
some heuristic arguments. It turns out that such approach cannot be directly extrapolated 
to parametrize the higher order contributions. The aim of this paper is to resolve this 
issue, providing a systematic procedure of constructing the slow-roll approximation
in loop quantum cosmology for the three major types of the slow-roll parameters 
discussed in literature.  
 
In case of the Hubble flow parameters \cite{Schwarz:2001vv,Schwarz:2004tz}  
the analysis has been already been partially performed in Refs. \cite{Zhu:2015xsa,
Zhu:2015owa}. The expansion is, however, not sufficient from the perspective of the 
potential shape reconstruction. Here, firstly the slow-roll expansion in terms of the 
Hubble slow-roll parameters is studied. Then, the Hubble flow parameters are 
expressed in terms of the Hubble slow-roll parameters. Finally, slow-roll expansion 
in terms of the potential parameters is discussed. In the calculations, the methodology 
developed at the level of standard inflationary cosmology \cite{Liddle:1994dx} is adopted. 

The starting point for our analysis is the modified Friedmann equation \cite{Ashtekar:2006wn}: 
\begin{equation}
H^2= \frac{\kappa}{3} \rho\left( 1- \frac{\rho}{\rho_c} \right),
\label{eq:H-kwant}
\end{equation}
where $\rho_c$ is a maximal allowed energy density and $\kappa := \frac{8\pi}{m^2_{\text{Pl}}}$. 
In our calculations, the energy density of the scalar field takes the classical form 
\begin{equation} \label{eq:rho-inflaton}
\rho = \frac {1}{2} \dot{\phi}^{2} + V(\phi),
\end{equation} 
and the field variable $\phi$ satisfies the Klein-Gordon equation:
\begin{equation}\label{eq:klein-gordon} 
\ddot\phi + 3H\dot\phi +  V_{, \phi }=0.
\end{equation}
Furthermore, we introduce the following dimensionless ratio
\begin{equation} \label{eq:delta}
\delta_H := \frac{\rho}{\rho_c} \leq 1,
\end{equation}
which parametrizes departure from the classical case, will be considered here as an 
expansion parameter, supplementary to the slow-roll parameters. Except the maximal 
value $\delta_H=1$, two other values of the parameter $\delta_H$ are distinguished. 
The first is $\delta_H=0$, which corresponds to the unmodified case, associated 
with taking  $\rho_c \rightarrow \infty$. Alternatively, $\delta_H\rightarrow 0 $ just when 
energy density is decreasing, which is an another way to recover the classical 
behavior. The second, more interesting value of the deformation parameter is 
$\delta_H=\frac{1}{2}$. As indicated by the analysis of the perturbative sector of 
inhomogeneities \cite{Cailleteau:2011kr,Bojowald:2011aa,Bojowald:2015gra}, the 
value of $\delta_H=\frac{1}{2}$ (for which $H=H_{\text{max}}$) corresponds to the 
signature change between Lorentzian and Euclidean.  Namely, in the deep quantum 
regime $\delta_H \in (1/2,1]$ the space-time signature is Euclidean, while when 
passing to the ``low'' energy density range  $\delta_H \in (0,1/2)$ the space-time
acquires the Lorentzian signature, as expected. Because of this effect, in practice, 
the region $\delta_H \in (0,1/2)$ will be most adequate from the perspective of 
studying the slow-roll approximation. 

Throughout this article we  consequently apply the Planck units, where $\hbar = c = 1$ and 
$G=1/m^{2}_{\text{Pl}}$ and where $m_{\text{Pl}}$ denotes the Planck mass. 

\section{Hubble slow-roll parameters}

Our goal is to construct systematic slow-roll approximation of the inflationary dynamics 
characterized by the condition 
\begin{equation} 
\label{eq:inflacja}
\ddot{a} > 0.
\end{equation}

A central role in the slow-roll approximation is played by the slow-roll parameters.
The parameters are defined such that the approximation is valid in the range of the absolute 
values of the slow-roll parameters being much smaller than unity. Among different possible
types of the slow-roll parameters, the Hubble slow-roll parameters and the potential 
slow-roll parameters \cite{Liddle:1994dx} are in common use. 

In this section, we focus on studying the slow-roll approximation for the cosmological
dynamics characterized by Eq. \ref{eq:H-kwant} using the first type of the slow-roll 
parameters. Despite the fact that the dynamics under consideration is deformed with 
respect to the standard one, in our analysis we keep the definitions of the Hubble slow-roll 
parameters in the standard form. Such kinematical approach allows us to compare 
results characterized by different dynamical equations. 

In case of the Hubble slow-roll parameters, the Hubble factor is considered as an 
explicit function of $\phi$, namely $H=H(\phi)$. The formulation is, therefore, reserved 
to the single field inflationary models. Following Ref. \cite{Liddle:1994dx} the parameters 
can be introduced as follows:
\begin{equation} \label{eq:epsilonh} 
\epsilon_{H} :=\frac{2}{\kappa}\left( \frac { H_{, \phi } }{H} \right)^{2},
\end{equation}
\begin{equation} \label{eq:etah} 
\eta_{H} :=\frac{2}{\kappa} \frac { H_{, \phi \phi} }{H}, 
\end{equation}
\begin{equation} \label{eq:xih} 
\xi_{H} :=\frac{2}{\kappa} \left( \frac { H_{, \phi }  H_{, \phi \phi \phi} }{H^2} \right)^{\frac{1}{2}},
\end{equation}
\begin{equation} \label{eq:sigmah}
\sigma_{H} := \frac{2}{\kappa} \left( \frac { (H_{, \phi })^2  H_{, \phi \phi \phi \phi} }{H^3} \right)^{\frac{1}{3}}.
\end{equation}
In general, the Hubble slow-roll parameters can be generated from the following formula:
\begin{equation} \label{eq:parametry-H-def}
{}^{n}\beta_{H} :=
\left\lbrace 
\begin{array}{l}
\epsilon_H \text{ for } n=0, \\
\frac{2}{\kappa} \left( \frac {( H_{, \phi} )^{n-1} \left(\frac{d^{n+1}}{d\phi^{n+1}}H \right)} 
{H^{n}} \right)^{\frac{1}{n}} \text{ for } n \geq 1 .
\end{array}
\right.
\end{equation}

The idea of slow-roll expansion is that in the inflationary phase, the slow-roll 
parameters satisfy the condition $|{}^{n}\beta_{H}|<1$, which allows for perturbative 
analysis. 

Let us verify this condition in case of $n=1$, corresponding to the parameter $\epsilon_H$. 
For this purpose, we have to differentiate Eq. \eqref{eq:H-kwant} with respect to time, 
using (\ref{eq:rho-inflaton}) and (\ref{eq:klein-gordon}). This results in
\begin{equation} \label{eq:H-dot}
\dot{H} = - \frac{\kappa}{2} \dot{\phi}^2 \left(1 - 2 \delta_H \right),
\end{equation}
which, with use of $\dot{H} =  H_{, \phi } \dot{\phi}$, can be further transformed to
\begin{equation} \label{eq:H-phi}
 H_{, \phi }  = - \frac{\kappa}{2} \dot{\phi} \left(1 - 2 \delta_H \right).
\end{equation}
Applying the obtained expression to definition (\ref{eq:epsilonh}), we find that 
\begin{equation}
\epsilon_H = - \frac{\dot{H}}{H^2}(1-2\delta_H)=\left(1- \frac{\ddot{a}}{aH^2}\right)(1-2\delta_H). 
\end{equation}
It is transparent that, in the domain $\delta_H <\frac{1}{2}$ the condition (\ref{eq:inflacja}) 
translates into $\epsilon_H<1$.

Our task now is to express $H^2$ up to the fourth order in both, the  slow-roll parameters and 
the parameter  (\ref{eq:delta}). For this purpose, we re-express (\ref{eq:H-phi}) to the form 
\begin{equation} \label{eq:dot-phi-H-phi}
\dot{\phi} = - \frac{2}{\kappa} \frac{ H_{, \phi } }{1 - 2 \delta_H}
\end{equation}
and apply it in the formula \eqref{eq:rho-inflaton}, obtaining 
\begin{equation} \label{eq:rho-H-phi}
\rho =  \frac{2( H_{, \phi } )^2}{\kappa^2 \left(1 - 2 \delta_H \right)^2} + V.
\end{equation}

In the next step, from \eqref{eq:H-kwant} and \eqref{eq:rho-H-phi} we obtain expression 
\begin{equation}\label{eq:H-kwant-2}
H^2 = \frac{\kappa}{3} V \frac{1 - \delta_H}{1 - \frac{1}{3} \epsilon_H \frac{1 - \delta_H}{\left(1 - 2 \delta_H \right)^2}},
\end{equation}
which, by expanding, gives us $H^2$ up to the fourth order in $\delta_H$ and the 
Hubble slow-roll parameter $\epsilon_H$:
\begin{eqnarray} 
\label{eq:H-kwant-epsilonH}
H^2 &=& \frac{\kappa}{3} V \left(1 + \frac{1}{3} \epsilon_H - \delta_H + \frac{1}{9} \epsilon_{H}^{2} + \frac{2}{3}\epsilon_H \delta_H 
+ \frac{1}{27} \epsilon_H^3 + \frac{5}{3} \epsilon_H \delta_H^2 + \frac{5}{9} \epsilon_H^2 
\delta_H \right. \nonumber \\
&+&  \left.  \frac{1}{81}  \epsilon_H^4+\frac{8}{27}\epsilon_H^3\delta_H+\frac{19}{9}\epsilon_H^2\delta_H^2+4\epsilon_H\delta_H^3+\dots  \right).  
\end{eqnarray}
In classical limit ($\delta_H = 0$), the result is coincides with the one derived in Ref. \cite{Liddle:1994dx}.

\section{Hubble flow parameters} \label{HFP}

A drawback of the Hubble and potential type slow-roll parameters is that the 
definitions are applicable for single inflation field models only. Therefore, it is 
desired to introduce definitions which work also for other types of matter contents, 
such as multi-component scalar field. 

In the course of inflation, the Hubble factor is a nearly constant (slowly decreasing) function
of a scale factor. The rate of change of the Hubble factor or its inverse (the horizon/Hubble
radius) as a function of the scale factor is a natural starting point for the definition 
of the the slow-roll parameters (in a way independent on the form of the matter content). 
This observation has been employed in the definition of the so-called 
``Hubble flow parameters,'' which are introduced with use of the recurrence relation 
\cite{Schwarz:2001vv,Schwarz:2004tz}:  
\begin{equation}
\epsilon_{n+1} \equiv \frac{d \ln \epsilon_{n} }{d \ln a} = \frac{1}{H} \frac{\dot{\epsilon}_{n}}{\epsilon_{n}}, 
\label{recurence1}
\end{equation}
together with the initial condition
\begin{equation}
\epsilon_0 \equiv \frac{H_0}{H}.
\label{initial1} 
\end{equation}

In the context of loop quantum cosmology, the Hubble flow parameters (up to the 
quadratic order) have been used in Refs. \cite{Zhu:2015xsa,Zhu:2015owa}, where 
quantum corrections to the inflationary spectra were analyzed. 

In what follows we analyze expressions for the first four Hubble flow parameters 
and relate them with the Hubble slow-roll parameters introduced in the previous 
section. Further application of the results of the Sec. \ref{PSRP} allows us to 
determine dependence of the Hubble slow-roll parameters on the potential 
slow-roll parameters. With use of such relation, the procedure of reconstruction 
of the potential function based on the observationally determined values of 
$\epsilon_{n}$ may be conducted. 

Application of the recurrence relation (\ref{recurence1}), together with the 
initial condition (\ref{initial1}), leads us to the following expressions:
\begin{eqnarray}
\epsilon_{1} &=& - \frac{\dot{H}}{H^2} =\frac{\epsilon_H}{1-2\delta_H}  
= \epsilon_H+2\epsilon_H\delta_H+4\epsilon_H\delta_H^2+8\epsilon_H \delta_H^3+\dots. \\
\epsilon_{2} &=& \frac{\ddot{H}}{\dot{H}H}-2\frac{\dot{H}}{H^2} 
= \frac{2}{(1-2\delta_H)}(\epsilon_H-\eta_H)-4\epsilon_H\delta_H 
\frac{(1-\delta_H)}{(1-2\delta_H)^3} \nonumber  \\
&=&2\epsilon_H-2 \eta_H-12\epsilon_H\delta_H^2-56\epsilon_H\delta_H^3-4\eta_H\delta_H-8\eta_H\delta^2_H-16\eta_H \delta_H^3+\dots \label{epsilon2} \\
\epsilon_{3} & =& \frac{1}{\ddot{H} H - 2 \dot{H}^2} \left(\dddot{H} - \frac{\ddot{H}^2}{\dot{H}} 
- 3 \frac{\ddot{H} \dot{H}}{H} + 4 \frac{\dot{H}^3}{H^2}
\right)\\
\epsilon_{3}\epsilon_{2}&=& 4 \epsilon_H^2-6 \eta _H \epsilon_H+2 \xi_H^2+8 \delta_H \epsilon _H^2
-12 \delta_H  \eta _H \epsilon_H+8 \delta_H  \xi_H^2+24 \delta_H^2 \xi_H^2 \nonumber  \\ 
&+&12 \delta_H^2 \eta _H \epsilon _H+24 \delta_H^2 \epsilon _H^2+\dots  \\
\epsilon_{4} & =& \left[ 
-H^4 \dot{H} \ddot{H}^2 \dddot{H}+H^4 \ddot{H}^4+H^3 \dot{H}^2 \left(H \ddot{H}
   \ddddot{H}-8 \ddot{H}^3-H \dddot{H}^2\right)+7 H^3 \dot{H}^3 \ddot{H} \right. \nonumber \\
&&    \left.  \dddot{H}-2 H^2 \dot{H}^4 \left(H \ddddot{H}-6 \ddot{H}^2\right)+2 H^2 \dot{H}^5
   \dddot{H}-26 H \dot{H}^6 \ddot{H}+16 \dot{H}^8
\right]/  \nonumber\\
&& \left[ H^2 \dot{H} \left(H \ddot{H}-2 \dot{H}^2\right) \left(H^2 \dot{H} \dddot{H}-H^2 \ddot{H}^2-3 H
   \dot{H}^2 \ddot{H}+4 \dot{H}^4\right)\right]  \\
\epsilon_{4}\epsilon_{3}\epsilon^2_{2} & =& 4 \eta _H^2 \xi _H^2+4 \eta _H \sigma _H^3
-4 \xi _H^4-52 \eta _H \epsilon _H^3+56 \eta _H^2 \epsilon _H^2-24 \eta _H^3 \epsilon _H \nonumber \\ 
&+&4 \xi _H^2\epsilon _H^2-4 \sigma _H^3 \epsilon _H+16 \epsilon _H^4+\dots
\end{eqnarray}

Methodology behind obtaining these formulas is the following: Firstly, we iteratively applied 
the recurrence relation (\ref{recurence1}) to express the Hubble flow parameters in terms 
of the time derivatives of the Hubble parameters. Secondly, we systematically replaced the 
time derivatives of the Hubble parameters by appropriate differentiations with respect to 
the scalar field. In particular, in the lowest order, the adequate formula is:
\begin{equation}
\dot{H} = \dot{\phi} H_{,\phi} = -\frac{2}{\kappa} \frac{ H^2_{, \phi } }{1 - 2 \delta_H},
\end{equation} 
where we used Eq. (\ref{eq:dot-phi-H-phi}). Finally, definitions of the Hubble slow-roll 
parameters were applied. While the procedure is rather straightforward, complexity of 
calculations dramatically increases together with the growth of $n$. 

As an example, let us discuss some of the steps in derivation of $\epsilon_2$. Here, we have to deal 
with the expression $\ddot{H}$, which with use of $\dot{H} = \dot{\phi} H_{,\phi}$, can 
be expressed as
\begin{equation}
\ddot{H} = \ddot{\phi}H_{,\phi} +\dot{\phi}^2 H_{,\phi\phi}.
\label{ddotHC}
\end{equation} 
Now, defining $C:=  -\frac{2}{\kappa} \frac{1}{(1-2 \delta_H)}$, the  Eq. (\ref{eq:dot-phi-H-phi})
can be written as $\dot{\phi}=CH_{,\phi}$, which after differentiation gives 
\begin{equation}
\ddot{\phi} = \dot{C} H_{,\phi}+C^2H_{,\phi}H_{,\phi\phi}. 
\label{ddotphiC}
\end{equation} 
The time derivative of the auxiliary function $C$ can be expressed as 
\begin{equation}
\dot{C} =4 \delta_H C^2\frac{(1-\delta_H)}{(1-2\delta_H)^2}\frac{H^2_{,\phi}}{H},    
\end{equation} 
which, when applied to Eq. \ref{ddotphiC}, and then by substituting  Eq. \ref{ddotphiC}
in Eq. \ref{ddotHC}, leads to the following formula:
\begin{equation}
\ddot{H} =  4 \delta_H C^2\frac{(1-\delta_H)}{(1-2\delta_H)^2}\frac{H^4_{,\phi}}{H}+
2C^2H^2_{,\phi}H_{,\phi\phi}.
\label{ddotHHubble}
\end{equation}

The equation (\ref{ddotHHubble}) can be now applied in the expression for $\epsilon_2$, 
written as a function of $H,\dot{H}$ and $\ddot{H}$. Finally, the definition of $C$ and the 
definitions of the slow-roll parameters $\eta_H$ and $\epsilon_H$ have to be used to obtain
Eq. (\ref{epsilon2}). 
 
\section{Potential slow-roll parameters} \label{PSRP}

The last type of the slow-roll parameters we are going to study here are the 
potential slow-roll parameters, which are roughly a measure of ``flatness'' 
of the potential functions. 

In the less formalized manner, the potential slow-roll parameters in the 
linear order were  a subject of investigations in Refs. 
\cite{Artymowski:2008sc,Ashtekar:2009mm,Mielczarek:2010bh}. 
In these referred articles, the potential slow-roll parameters, $\epsilon_V$ 
and $\eta_V$ were defined on the basis of heuristic arguments:
\begin{eqnarray}
\overline{\epsilon_{V}} &:=& \frac{1}{2\kappa}\left( \frac { V_{, \phi } }{V} \right)^{2} 
\frac{1}{1-\frac{V}{\rho_c}}, \\
\overline{\eta_{V}} &:=& \frac{1}{\kappa}\frac{ V_{, \phi \phi} }{V}\frac{1}{1 
- \frac{V}{\rho_c}},
\end{eqnarray}
where we introduced overlies to distinguish the definitions from the ones discussed 
here. While form of $\overline{\epsilon_{V}}$ and $\overline{\eta_{V}}$ can be 
deduced from the slow-roll condition (\ref{eq:inflacja}) combined with the 
modified Friedmann equation (\ref{eq:H-kwant}), the same arguments cannot provide us
definitions for the further potential slow-roll parameters. Therefore, it is reasonable to use 
standard rather than the modified (by the factor  $1/(1-V/\rho_c)$) definitions 
of the slow-roll parameters. This is also due to the fact that the role of the potential 
slow-roll parameters is to characterize a potential function, independently on the 
dynamics. The factor  $\frac{1}{1-V/\rho_c}$ is explicitly a sign of the deformed 
dynamics present in the model.

For our purposes, only the first four potential slow-roll parameters are relevant \cite{Liddle:1994dx}:
\begin{equation} 
\label{eq:epsilonv} 
\epsilon_{V} :=\frac{1}{2\kappa}\left( \frac { V_{, \phi } }{V} \right)^{2},
\end{equation}
\begin{equation} 
\label{eq:etav} 
\eta_{V} :=\frac{1}{\kappa} \frac { V_{, \phi \phi} }{V} ,
\end{equation}
\begin{equation}
\label{eq:xiv} \xi_{V} :=\frac{1}{\kappa} \left( \frac { V_{, \phi }  V_{, \phi \phi \phi} }{V^2} \right)^{\frac{1}{2}},
\end{equation}
\begin{equation} \label{eq:sigmav}
\sigma_{V} := \frac{1}{\kappa} \left( \frac { (V_{, \phi })^2  V_{, \phi \phi \phi \phi} }{V^3} \right)^{\frac{1}{3}}.
\end{equation}
In general, the following definition is satisfied:
\begin{equation} \label{eq:parametry-V-def}
{}^{n}\beta_{V} :=
\left\lbrace 
\begin{array}{l}
\epsilon_V \text{ for } n=0, \\
\frac{1}{\kappa} \left( \frac {( V_{, \phi} )^{n-1} \left(\frac{d^{n+1}}{d\phi^{n+1}}V \right)} {V^{n}} \right)^{\frac{1}{n}} \text{ for } n \geq 1 .
\end{array}
\right.
\end{equation}

Furthermore, the following derivatives of the potential slow-roll parameters with respect to $\phi$ are useful:
\begin{eqnarray} 
\label{eq:pochodna1-epsilonV}
\frac{d}{d \phi} \epsilon_V &=& 2 \kappa \frac{V}{V_{, \phi}} (- 2 \epsilon_V^2 + \epsilon_V \eta_V) \\
\label{eq:pochodna1-etaV}
\frac{d}{d \phi} \eta_V &=& \kappa \frac{V}{V_{, \phi}} \left(- 2 \epsilon_V \eta_V + \xi_V^2 \right), \\
 \label{eq:pochodna1-xiV}
\frac{d}{d \phi} \xi_V^2 &=& \kappa \frac{V}{V_{, \phi}} \left(- 4 \epsilon_V \xi_V^2 + \eta_V \xi_V^2 + \sigma_V^3  \right), \\
\label{eq:pochodna2-epsilonV}
\frac{d^2}{d \phi^2} \epsilon_V &=& 2 \kappa  \left(\eta _V^2+\epsilon _V^2 \left(2 \kappa ^2 \left(\xi _V^2-\eta _V^2\right)+6\right)+4 \kappa ^2 \eta _V \epsilon _V^3-6 \eta _V \epsilon _V\right), \\
\label{eq:pochodna2-etaV}
\frac{d^2}{d \phi^2} \eta_V &=& \kappa  \left(-2 \eta _V^2-\xi _V^2+4 \kappa ^2 \epsilon _V^2 \left(\eta _V-\xi _V\right) \left(\eta _V+\xi _V\right)+2 \epsilon _V 
\left(2 \eta _V+\kappa ^2 \sigma _V^3\right)\right), \\
\label{eq:pochodna2-xiV}
\frac{d^2}{d \phi^2} \xi_V^2 &=& \kappa \left(-6 \eta _V \xi _V^2+\sigma _V^2+2 \epsilon _V \left(\kappa ^2 \left(\eta _V \sigma _V^2 \left(3 \sigma _V-1\right)
+\tau _V^4\right)+\kappa ^2 \xi _V^4+6 \xi _V^2\right) \right. \nonumber \\
&-& 20  \left. \kappa ^2 \sigma _V^3 \epsilon _V^2\right), \\
\label{eq:pochodna-delta-1}
\delta_{H, \phi} &=& \frac{3}{\kappa} \frac{\frac{d}{d\phi} H^2}{\rho_c (1- 2 \delta_H)},\\
\label{eq:pochodna-delta-2}
\delta_{H, \phi \phi} &=& \frac{3}{\kappa} \frac{\frac{d^2}{d\phi^2} H^2}{\rho_c (1- 2 \delta_H)}
+ \frac{6}{\kappa} \frac{\delta_{H, \phi} \frac{d}{d \phi} H^2}{\rho_c (1- 2 \delta_H)^2}.
\end{eqnarray}

The method which will be presented below allows to determine solution of \eqref{eq:H-kwant}
to an arbitrary order in the potential slow-roll parameters. 

In order to make the calculations easier, we introduce an auxiliary parameter:
\begin{equation} \label{eq:f} 
f := \frac{ H_{, \phi } }{H} = \frac{\frac{d}{d \phi} H^2}{2 H^2}
\end{equation}
and write the first Hubble slow-roll parameter in terms of it:
\begin{equation} \label{eq:epsilonH-f}
\epsilon_H = \frac{2}{\kappa} f^2.
\end{equation}

Now, we will proceed iteratively, employing the following steps:
\begin{enumerate}
\item Start from $H^2$ to $n$-th order in potential slow-roll parameters and in parameter 
$\delta_H$ (the approximation for $n = 0$ is $H^2 = \frac{\kappa}{3} V$). 
\item Using this approximation, expressions for derivatives of potential slow-roll 
parameters \eqref{eq:pochodna1-epsilonV}-\eqref{eq:pochodna2-xiV} and values 
of derivatives of $\delta_H$ to the $(n-1)$-th order obtain derivatives of $\delta_H$ 
to the $n$-th order from \eqref{eq:pochodna-delta-1} and \eqref{eq:pochodna-delta-2}. 
The first non-zero values are for $n=2$: $\delta_{H, \phi} \simeq \frac{V_{, \phi}}{V} 
\delta_H$, $\delta_{H, \phi \phi} \simeq \kappa \eta_V \delta_H$.
\item Calculate an approximate value of $f$ from \eqref{eq:f} and then $\epsilon_H$ 
up to the $(n+1)$-th order from \eqref{eq:epsilonH-f}.
\item Finally, calculate $H^2$ to the $(n+1)$-th order from \eqref{eq:H-kwant-epsilonH}.
\end{enumerate}

The above procedure is valid because of the following:
\begin{lemat} \label{lemat2}
For any expression which is a product of slow-roll parameters and quantities independent 
on $\phi$, and it is of the order $k \geq 0$ in slow-roll parameters, the derivative of this 
expression with respect to $\phi$ is at least of the order $k$. 
\end{lemat}

\begin{proof}
We will omit numerical constants, writing ''$\sim$''. For $k = 0$ after differentiation 
we obtain $0$, so the lemma is satisfied. For $k = 1$ there are two cases dependent 
on value of ''$n$'' in ${}^{n}\beta_{V}$. For $n = 0$:

\begin{equation} \label{eq:pochodna_epsilon_V}
\frac{d}{d \phi} {}^{0}\beta_{V} = \epsilon_{V, \phi} 
= \frac{ V_{, \phi } }{V} \left( \eta_V - 2 \epsilon_V \right),
\end{equation}
which is of order $\frac{3}{2}$ in slow-roll parameters. Next, for $n \geq 1$:
\begin{eqnarray}
\frac{d}{d \phi} {}^{n}\beta_{V} &\sim& V^{-2}  V_{, \phi }  ( V_{, \phi } )^{\frac{n-1}{n}} 
\left(\frac{d^{n+1}}{d\phi^{n+1}}V \right)^{\frac{1}{n}} 
+ V^{-1} ( V_{, \phi } )^{\frac{n-1}{n} - 1}  V_{, \phi \phi}  \left(\frac{d^{n+1}}{d\phi^{n+1}}V \right)^{\frac{1}{n}}  \nonumber  \\
&+& V^{-1} ( V_{, \phi } )^{\frac{n-1}{n}} \left(\frac{d^{n+2}}{d\phi^{n+2}}V \right) \left(\frac{d^{n+1}}{d\phi^{n+1}}V \right)^{\frac{1}{n} - 1} \nonumber \\
&\sim& {}^{n}\beta_{V} \frac{ V_{, \phi } }{V}
+ {}^{n}\beta_{V} \frac{ V_{, \phi \phi} }{V_{, \phi } }
+ \left({}^{n+1}\beta_{V} \right)^{n+1}
\left({}^{n}\beta_{V}\right)^{1 - n} 
\frac{V}{ V_{, \phi }},
\end{eqnarray}
which is also of order $\frac{3}{2}$ in slow-roll parameters. Let us consider expression of arbitrary order $k \geq 2$:
\begin{equation}
\left({}^{n_1}\beta_{V} \right)^{k_1} \cdot \left({}^{n_2}\beta_{V} \right)^{k_2} \cdot \ldots \cdot \left({}^{n_s}\beta_{V} \right)^{k_s},
\end{equation}
where $k_1 + k_2 + \ldots + k_s = k$. 
We calculate:
\begin{equation}
 \label{eq:dluga-pochodna-slow-roll}
\begin{split}
& \frac{d}{d \phi} \left[ \left({}^{n_1}\beta_{V} \right)^{k_1} \cdot \left({}^{n_2}\beta_{V} \right)^{k_2} \cdot \ldots \cdot \left({}^{n_{s-1}}\beta_{V} \right)^{k_{s-1}} \cdot \left({}^{n_s}\beta_{V} \right)^{k_s} \right] = \\
&=\left\lbrace \frac{d}{d \phi} \left[ \left({}^{n_1}\beta_{V} \right)^{k_1} \right]  \right\rbrace \cdot \left[  \left({}^{n_2}\beta_{V} \right)^{k_2} \cdot \ldots \cdot \left({}^{n_s}\beta_{V} \right)^{k_s} \right]  \\
&
+ \left\lbrace \frac{d}{d \phi} \left[ \left({}^{n_2}\beta_{V} \right)^{k_2} \right] \right\rbrace \cdot \left[  \left({}^{n_1}\beta_{V} \right)^{k_1} \cdot \ldots \cdot \left({}^{n_s}\beta_{V} \right)^{k_s} \right] 
+ \ldots  \\
& + \left\lbrace \frac{d}{d \phi} \left[ \left({}^{n_s}\beta_{V} \right)^{k_s} \right] \right\rbrace \cdot \left[ \left({}^{n_1}\beta_{V} \right)^{k_1} \cdot \left({}^{n_2}\beta_{V} \right)^{k_2}  \cdot \ldots \cdot \left({}^{n_s-1}\beta_{V} \right)^{k_s-1} \right] .
\end{split}
\end{equation}
Notice that:
\begin{equation} \label{eq:krotka-pochodna-slow-roll}
\frac{d}{d \phi} \left[ \left({}^{n_i}\beta_{V} \right)^{k_i} \right] 
= k_i  \left({}^{n_i}\beta_{V} \right)^{k_i - 1} \cdot 
\frac{d}{d \phi} \left[ \left({}^{n_i}\beta_{V} \right) \right],
\end{equation}
which is of order $\frac{3}{2} + k_i - 1 = k_i + \frac{1}{2}$ in slow-roll 
parameters. By virtue of \eqref{eq:krotka-pochodna-slow-roll} every 
element of sum \eqref{eq:dluga-pochodna-slow-roll} is of order $\frac{3}{2}
+\left( \sum_{i = 1}^{s} k_i \right) -1 = \frac{1}{2} + k > k$, which completes the 
proof. \end{proof}

The lemma guarantees that when we differentiate a given expression, the result is 
of at least the same order in slow-roll parameters, so we do not lose any terms.

In order to better visualize technology of the method, let us explicitly calculate (up to the 
second order) expressions for $\epsilon_H$ and $H^2$ in terms of the potential slow-roll
parameters. 

The 0th order expression for $H^2$ is $H^2=\frac{\kappa}{3} V$.  Applying 
this to the definition of $f$ we obtain $f=\frac{1}{2}  \frac{V_{,\phi}}{V}$, 
based on which $\epsilon_H = \frac{2}{\kappa} f^2=\frac{1}{2\kappa}
=\frac{1}{2\kappa}\left( \frac { V_{, \phi } }{V} \right)^{2} = \epsilon_V$. 
Applying this to (\ref{eq:H-kwant-epsilonH}), we obtain the 1st order 
expression:
\begin{equation}
H^2 = \frac{\kappa}{3} V \left(1 + \frac{1}{3} \epsilon_V - \delta_H\right).  
\label{H1V}
\end{equation}
Now, let us repeat the procedure with use of  (\ref{H1V}) in expression 
(\ref{eq:f}). We obtain  $f=\frac{1}{2}  \frac{V_{,\phi}}{V}\left(1-\frac{2}{3}\epsilon_V
+\frac{1}{3}\eta_V-\delta_H \right) $, which leads to 
$\epsilon_H = \epsilon_V - \frac{4}{3} \epsilon_V^2+ \frac{2}{3} \epsilon_V \eta_V- 2 \epsilon_V \delta_H $. By substituting this expression in (\ref{eq:H-kwant-epsilonH}),
we find that, up to the second order:
\begin{equation}
H^2 = \frac{\kappa}{3} V \left(1 + \frac{1}{3} \epsilon_V  - \delta_H
- \frac{1}{3} \epsilon_{V}^{2} + \frac{2}{9} \epsilon_V \eta_V \right).  
\label{H2V}
\end{equation}
  
Following this procedure, we find the 4th order expression of the Hubble slow-roll 
parameter 
\begin{equation}
\begin{split}
\epsilon_H & = \epsilon_V  
- \frac{4}{3} \epsilon_V^2+ \frac{2}{3} \epsilon_V \eta_V- 2 \epsilon_V \delta_H 
+ \frac{32}{9} \epsilon_V^3 - \frac{10}{3} \epsilon_V^2 \eta_V + \frac{5}{9} \epsilon_V \eta_V^2 + \frac{2}{9} \epsilon_V \xi_V^2  + \epsilon_V^2 \delta_H   \\
& - \frac{2}{3} \epsilon_V \eta_V \delta_H - \epsilon_V \delta_H^2 
- \frac{44}{3} \epsilon_V^4 + \frac{530}{27} \epsilon_V^3 \eta_V  - \frac{62}{9} \epsilon_V^2 \eta_V^2 + \frac{14}{27} \epsilon_V \eta_V^3  - \frac{16}{9} \epsilon_V^2 \xi_V^2 \\
& + \frac{2}{3} \epsilon_V \eta_V \xi_V^2 + \frac{2}{27} \epsilon_V \sigma_V^3 + \frac{13}{18} \epsilon_V^3 \delta_H- \frac{4}{9} \epsilon_V^2 \eta_V \delta_H + \frac{8}{3} \epsilon_V^2 \delta_H^2- \frac{4}{3} \epsilon_V \eta_V \delta_H^2 
+ \dots
\end{split}
\end{equation}
which by applying to  (\ref{eq:H-kwant-epsilonH}), leads to the following 4th 
order expression:
\begin{equation}\label{eq:H-kwant-wynik}
\begin{split}
H^2 &= \frac{\kappa}{3} V \left(
1 + \frac{1}{3} \epsilon_V  - \delta_H
- \frac{1}{3} \epsilon_{V}^{2} + \frac{2}{9} \epsilon_V \eta_V 
+ \frac{25}{27} \epsilon_V^3 - \frac{26}{27} \epsilon_V^2 \eta_V  + \frac{5}{27} \epsilon_V \eta_V^2 + \frac{2}{27} \epsilon_V \xi_V^2  \right. \\
&- \left. \frac{4}{9} \epsilon_V^2 \delta_H + \frac{2}{9} \epsilon_V \eta_V \delta_H 
- \frac{328}{81} \epsilon_V^4 + \frac{460}{81} \epsilon_V^3 \eta_V - \frac{172}{81} \epsilon_V^2 \eta_V^2 + \frac{14}{81} \epsilon_V \eta_V^3 - \frac{44}{81} \epsilon_V^2 \xi_V^2 \right. \\
& \left. + \frac{2}{9} \epsilon_V \eta_V \xi_V^2 + \frac{2}{81} \epsilon_V \sigma_V^3
+ \frac{109}{54} \epsilon_V^3 \delta_H - \frac{56}{27} \epsilon_V^2 \eta_V \delta_H + \frac{10}{27} \epsilon_V \eta_V^2 \delta_H + \frac{4}{27} \epsilon_V \xi_V^2 \delta_H \right. \\
& \left. - \frac{5}{9} \epsilon_V^2 \delta_H^2 + \frac{2}{9} \epsilon_V \eta_V \delta_H^2
+ \dots \right).
\end{split}
\end{equation}

In classical limit ($\delta_H = 0$), the result is consistent with the one derived in Ref. \cite{Liddle:1994dx}.

\section{An attractor}

The slow-roll evolution is typically realized by the phase-space trajectories
approaching attractor of the dynamics. Expect for some special cases, finding 
expression for the attractor solution is a difficult task. One of them is 
the asymptotic attractor for the classical FRW dynamics with a single 
massive scalar field. In that case the attractors are asymptotically (far from 
the centre of the coordinate system) given by two trajectories which are 
parallel to the $\dot{\phi}=0$ axis at the $(\phi,\dot{\phi})$ plane.  In this case, 
the sign of  $\frac{d \dot{\phi}}{dt}$ depends on whether a trajectory is 
approaching the attractor either from below or from above. The transition 
line between the regions of different signs of $\frac{d \dot{\phi}}{dt}$ defines 
the attractor condition: 
\begin{equation} 
\frac{d \dot{\phi}}{dt}=-3H\dot\phi-V_{, \phi } \approx 0, 
\label{attract_mass}
\end{equation}
where we used the Klein-Gordon equation (\ref{eq:klein-gordon}). Under some 
conditions, validity of the attractor condition (\ref{attract_mass}) can be extrapolated 
to different models. Here, we address this issue in the context of the model being 
a subject of this paper. For this purpose, let us rewrite the Klein-Gordon equation 
(\ref{eq:klein-gordon}) in the following form:
\begin{equation}
\underbrace{\frac{\ddot{\phi}}{H\dot{\phi}}}_{:=A}+\underbrace{3+\frac{V_{,\phi}}{H\dot{\phi}}}_{:=B}=0.
\end{equation}
The attractor condition (\ref{attract_mass}) is satisfies if both $A$ and $B$ factors are
independently sufficiently small, $|A| \ll 1$ and $|B| \ll 1$. With use of Eq. \ref{ddotphiC}
one can find exact expression for the factor $A$  as a function of slow-roll expansion parameters:
\begin{equation}
A = -4\delta_H\epsilon_{\epsilon_H}\frac{(1-\delta_H)}{(1-2\delta_H)^3}-\eta_H \frac{1}{(1-2\delta_H)}
= -\eta_H +\mathcal{O}(2). 
\end{equation}
The condition $|A| \ll 1$ is, therefore, satisfied in the slow-roll regime. On the other hand, 
leading contribution to $B$ can be written as 
\begin{equation}
B =\eta_H-\epsilon_{H}+\mathcal{O}(2),  
\end{equation}
which also satisfies the condition  $|B| \ll 1$ in the slow-roll regime. One can, therefore, 
conclude that in the slow-roll regime, the approximate attractor condition (\ref{attract_mass}) 
is satisfied also in case when the holonomy corrections are present. The value of the 
$\delta_H$ parameter has to be, however, sufficiently small.    

With use of Eq. (\ref{attract_mass}) and expression on $H(\phi)$ given by Eq. \ref{eq:H-kwant-wynik}, 
the attractor equation can be written as:  
\begin{equation} \label{attractorEq} 
\begin{split}
\dot{\phi}  \simeq  - \frac{V_{,\phi}}{3 H(\phi)} &= \mp \sqrt{\frac{2}{3} V \epsilon_V}
\left(1-\frac{\epsilon _V}{6}+\frac{\delta _H}{2}
+\frac{5 \epsilon _V^2}{24}-\frac{\epsilon _V \eta _V}{9}-\frac{\epsilon _V \delta _H}{4}+\frac{3 \delta _H^2}{8}
 -\frac{241 \epsilon _V^3}{432} \right. \\
& \left.+\frac{29}{54} \epsilon _V^2 \eta _V-\frac{5}{54} \epsilon _V \eta _V^2-\frac{1}{27} \epsilon _V \xi
   _V^2+\frac{83}{144} \epsilon _V^2 \delta _H-\frac{5}{18} \epsilon _V \eta _V \delta _H-\frac{5}{16} \epsilon _V \delta _H^2  \right. \\
& \left.
   +\frac{5 \delta_H^3}{16}
   + \frac{299 \epsilon _V^4}{128}  -\frac{2047}{648} \epsilon _V^3 \eta _V+\frac{365}{324} \epsilon _V^2 \eta _V^2-\frac{7}{81} \epsilon _V \eta
   _V^3+\frac{47}{162} \epsilon _V^2 \xi _V^2 \right. \\
& \left.-\frac{1}{9} \epsilon _V \eta _V \xi _V^2 
   -\frac{1}{81} \epsilon _V \sigma _V^3 
   -\frac{1783}{864} \epsilon _V^3 \delta _H 
   +\frac{211}{108} \epsilon _V^2 \eta _V \delta _H-\frac{35}{108} \epsilon _V \eta _V^2 \delta _H \right. \\
& \left. -\frac{7}{54}
   \epsilon _V \xi _V^2 \delta _H+\frac{637}{576} \epsilon _V^2 \delta _H^2-\frac{35}{72} \epsilon _V \eta _V \delta _H^2
   -\frac{35}{96}
   \epsilon _V \delta _H^3+\frac{35 \delta _H^4}{128}+\dots \right). 
\end{split}
\end{equation}

\section{A case study}

The purpose of this section is to verify validity of the obtained approximations using the 
example of the scalar field with quadratic potential function:  
\begin{equation}
V(\phi) = \frac{1}{2} m^2 \phi^2.  
\label{MassPot}
\end{equation}
For such potential only the first two potential-type slow-roll parameters are non-vanishing:
\begin{equation}
\epsilon_V = \eta_V =  \frac{2}{\kappa} \frac{1}{\phi^2}. 
\label{MassSR}
\end{equation} 

In case of the quadratic potential, it is convenient to work with the variables \cite{Mielczarek:2010bh}:
\begin{equation}
x := \frac{m\phi}{\sqrt{2\rho_c}}, \ \ \text{and} \ \  y := \frac{\dot{\phi}}{\sqrt{2\rho_c}},
\label{xy}
\end{equation} 
which allow us to reformulate expression for the energy density into the form:
\begin{equation}
x^2+y^2 = \frac{\rho}{\rho_c} \leq 1.
\label{circcond} 
\end{equation} 

In the $x-y$ parametrization, the phase trajectories of the scalar field are 
confined in the unit circle, where the boundary corresponds to the bounce, 
$\rho=\rho_c$. 

The slow-roll approximation is valid up to $\epsilon_V  \approx 1$, which can be considered as 
the condition for termination of the inflationary period. This condition, for the massive potential case 
studied in this section, can be translated into an adequate range of $x$ at which the slow-roll 
approximation is valid. Namely, combining $\epsilon_V  \approx 1$ together with (\ref{circcond}), 
we find that $|x| \in [x_{\text{min}},1]$, where $x_{\text{min}} = 
\frac{1}{\sqrt{8\pi}}\left( \frac{m m_{\text{Pl}}}{\rho_c} \right) $.

Applying Eq. \ref{MassPot} and Eqs. \ref{MassSR}  to expression (\ref{attractorEq}) 
we obtain the following attractor equation:
\begin{equation}
y =  \mp \frac{1}{2\sqrt{6 \pi }}\left(  \frac{m m_{\text{Pl}}}{\sqrt{\rho_c}} \right) 
\left(1- \frac{1}{6} \frac{1}{8\pi}\left(  \frac{m^2 m^2_{\text{Pl}}}{\rho_c} \right)\frac{1}{x^2}+
\frac{1}{2}\left(x^2+\frac{1}{24\pi} \frac{m^2 m^2_{\text{Pl}}}{\rho_c}\right) +\dots \right),
\label{attractEq}
\end{equation}
where, for clarity, we presented only the leading ($\mathcal{O}(\epsilon_V,\eta_H)$) contributions. 
Derivation of the formula (\ref{attractEq}) requires expressing $\delta_H$ as a function the 
potential slow-roll parameters and the potential $V$. Consequently, with use of Eq. \ref{eq:dot-phi-H-phi} 
one can find that: 
\begin{equation}
\begin{split}
\delta_H  =  \frac{V}{\rho_c} & \left[
1+\frac{1}{3}\epsilon_V  
-\frac{\epsilon _V^2}{3}+\frac{2 \epsilon _V \eta _V}{9}+\frac{\epsilon _V \delta _H}{3}
+ \frac{25 \epsilon _V^3}{27}-\frac{26}{27} \epsilon _V^2 \eta _V+\frac{5}{27} \epsilon _V \eta _V^2+\frac{2}{27} \epsilon _V \xi
   _V^2 \right. \\
& \left.   
   -\frac{7}{9} \epsilon _V^2 \delta _H+\frac{4}{9} \epsilon _V \eta _V \delta _H+\frac{1}{3} \epsilon _V \delta _H^2
-\frac{109 \epsilon _V^4}{27}+\frac{460}{81} \epsilon _V^3 \eta _V-\frac{172}{81} \epsilon _V^2 \eta _V^2+\frac{14}{81} \epsilon _V \eta
   _V^3 \right. \\
& \left.   
 -\frac{44}{81} \epsilon _V^2 \xi _V^2+\frac{2}{9} \epsilon _V \eta _V \xi _V^2+\frac{2}{81} \epsilon _V \sigma _V^3+\frac{53}{18}
   \epsilon _V^3 \delta _H-\frac{82}{27} \epsilon _V^2 \eta _V \delta _H+\frac{5}{9} \epsilon _V \eta _V^2 \delta _H \right. \\
& \left.   
 +\frac{2}{9} \epsilon _V
   \xi _V^2 \delta _H-\frac{4}{3} \epsilon _V^2 \delta _H^2+\frac{2}{3} \epsilon _V \eta _V \delta _H^2+\frac{1}{3} \epsilon _V \delta _H^3+\dots
 \right]. 
\end{split}
\end{equation}

In Fig. \ref{PP} we compare the attractor trajectory (\ref{attractEq}), calculated at different 
orders of approximation, with representative dynamical trajectories obtained numerically. 
\begin{figure}[ht!]
\centering
\includegraphics[width=12cm,angle=0]{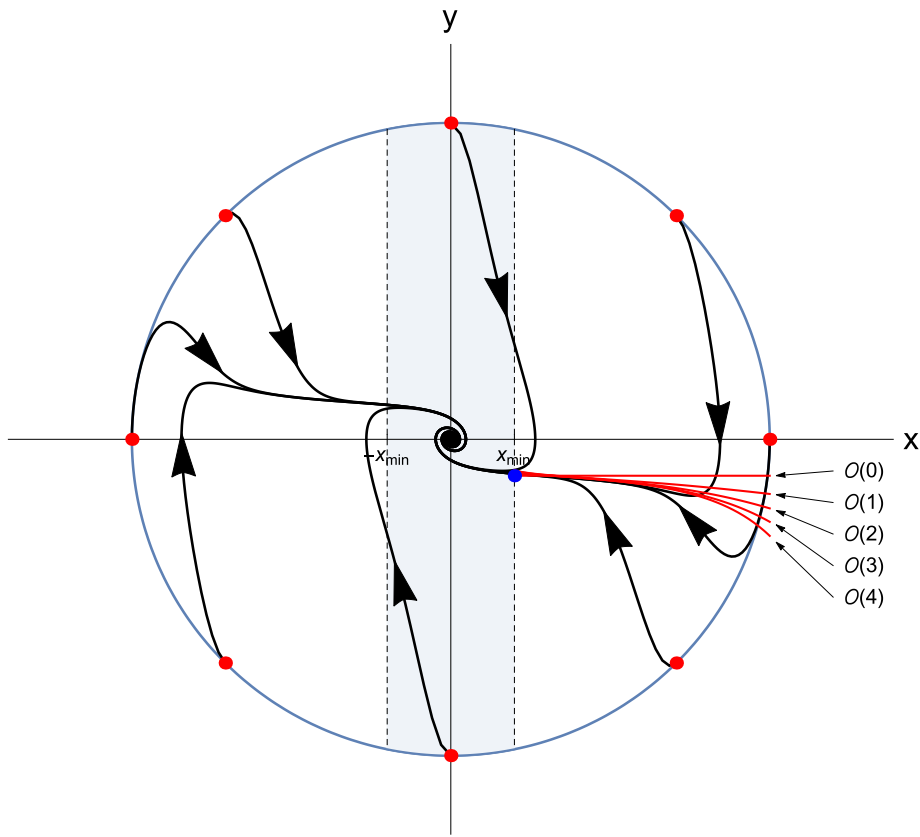}
\caption{Exemplary dynamical trajectories for the model with a massive scalar field.
Here, we fixed $m=m_{\text{Pl}}$ and $\rho_c=m^4_{\text{Pl}}$. The red points at the boundary 
represent the values at which the initial conditions have been imposed. In the shadowed region 
the slow-roll approximation is violated, $\epsilon_V>1$. The blue point at $\left( x_{\text{min}}, 
- \frac{1}{2\sqrt{6 \pi }}\left(\frac{m m_{\text{Pl}}}{\sqrt{\rho_c}} \right)  \right)$ indicates 
beginning of the  range of applicability of the slow-roll approximations discussed in this article. 
The trajectories starting from the points $(1,0)$ and $(0,-1)$ are the attractors of the dynamics.} 
\label{PP}
\end{figure}

Comparison of the analytical prediction with the numerical results confirms validity of the 
obtained formulas in the range from $x \approx x_{\text{min}}$ to $x \approx  0,7$.
For $x<x_{\text{min}}$ the slow-roll conditions are broken. On the other hand, at $x>0,7$
the slow-roll conditions are still valid while the approximation significantly departs
from the attractor trajectory (the trajectory which begins the point $(1,0)$). An explanation 
of this feature lies in limited range of applicability of the attractor condition, Eq. \ref{attract_mass}. 
Furthermore, the observed departure coincides with the value at which the Hubble factor 
takes its maximal value ($\rho=\rho_c/2$), which is $x=\sqrt{2}/2 \approx0,707$ for $y=0$. 
However, the region $x>\sqrt{2}/2$ for $y\rightarrow 0$ corresponds to the Euclidean sector 
of the theory, for which analysis of the slow-roll approximation may have limited practical 
relevance. 

In the presented example, we intentionally consider the case of large inflaton mass $m=m_{\text{Pl}}$.
Bigger the inflaton mass, shorter the inflationary period and less distinct the attractor solution. 
In such case differences between discussed orders of approximation become more significant, 
making comparison of the expressions more transparent. Decreasing the mass drives the attractor 
closer to the $y=0$ line, and improving accuracy of the approximation. Already a one order of 
magnitude decrease of the mass value leads to a significant improvement of the accuracy of 
the approximation, even beyond $x \approx  0,7$. 
 
\section{Summary}

In this article we have constructed systematic slow-roll expansion for the loop quantum 
cosmology model with a single scalar field. The approximations were studied in case of 
the Hubble type, horizon flow type and potential type slow-roll parameters. In our approach,
standard definitions of the slow-roll parameters were used. The developed methodology
has been employed to derive expressions for useful cosmological quantities up to the fourth 
order in both the slow-roll parameters and the parameter $\delta_H$. Validity of the obtained 
results was confirmed by confrontation with numerically determined trajectories for the 
model with the quadratic potential. It has been shown that the obtained analytical approximation
for the attractor trajectory converges with the numerical results in the Lorentzian domain 
($\rho<\rho_c/2$) if the slow-roll conditions are satisfied. 

The primary goal of the performed studies was to construct analytical approximation of the 
inflationary background dynamics in loop quantum cosmology, which is necessary for the 
further studies of generation of inhomogeneities. The corrections to inflation due to the 
presence of $\delta_H$ are very subtle. Therefore, any attempt to confront the predictions 
with cosmological data will require taking into account approximations of sufficiently high order. 
In Appendix we collected some additional formulas, which may be useful in the analysis 
of the perturbative inhomogeneities.  

\section*{Acknowledgements}

This work is supported by the National Science Centre (Poland), \\
projects DEC-2013/09/B/ST2/03455 and DEC-2014/13/D/ST2/01895. 

\section*{Appendinx}

The results obtained in this article may be used to derive slow-roll approximations 
for various cosmologically relevant quantities. Here, the expression for the 
scale factor $a$ as well as $\frac{a'}{a}$ and $\frac{a''}{a}$ are discussed. All of these 
functions find application in the studies of cosmological inhomogeneities. 

Let us define the conformal time, which can be expressed as follows
\begin{equation} 
\label{eq:def-czas-konforemny}
\begin{split}
\tau & := \int \frac{dt}{a}
= \int \frac{d\tau}{d\mathcal{H}} \mathcal{H}^2 \frac{d\mathcal{H}}{\mathcal{H}^2}
= - \int \frac{d\tau}{d\mathcal{H}} \mathcal{H}^2 d \left(\frac{1}{\mathcal{H}} \right)
= - \int \frac{1}{1 + \frac{\dot{H}}{H^2}} d \left(\frac{1}{\mathcal{H}} \right) ,
\end{split}
\end{equation}
where $\mathcal{H} := aH$. Integration of the last integral in (\ref{eq:def-czas-konforemny}) by parts 
leads to 
\begin{equation} \label{eq:def-czas-konforemny-2}
\tau = - \frac{1}{1 + \frac{\dot{H}}{H^2}} \frac{1}{\mathcal{H}} \left( 1 
- \frac{2 \left( \frac{\dot{H}}{H^2} \right)^2 - \frac{\ddot{H}}{H^3}}{\left( 1 + \frac{\dot{H}}{H^2} \right)^2} \right)^{-1}.
\end{equation}

Dividing \eqref{eq:H-dot} by $H^2$ we obtain
\begin{equation} \label{eq:dotH-nad-H2}
\frac{\dot{H}}{H^2} 
= - \frac{\kappa}{2} \frac{\dot{\phi}^2 (1 - 2\delta_H) }{H^2}
= - \frac{1}{2 \kappa (1 - 2\delta_H )} \left( \frac{\frac{d}{d \phi} H^2 }{H^2} \right)^2.
\end{equation}

A similar procedure leads to
\begin{equation} \label{eq:ddotH-nad-H3}
\frac{\ddot{H}}{H^3} = \frac{\delta_{H, \phi}}{\kappa^2 (1 - 2 \delta_H)^3} 
\left( \frac{\frac{d}{d \phi} H^2 }{H^2}\right)^3
+ \frac{\left( \frac{d}{d \phi} H^2 \right)^2  
\left( \frac{d^2}{d \phi^2} H^2 \right)}{\kappa^2 (1 - 2\delta_H )^2 H^6}
- \frac{1}{2 \kappa^2 (1 - 2\delta_H )^2} \left( \frac{\frac{d}{d \phi} H^2}{H^2} \right)^4.
\end{equation}

Combining \eqref{eq:def-czas-konforemny-2}, \eqref{eq:dotH-nad-H2} and 
\eqref{eq:ddotH-nad-H3} one can find the following expression for the scale factor: 
\begin{equation}
\begin{split}
a \simeq &  - \frac{1}{H \tau} 
\left( 1 + \epsilon_V 
+ \frac{11 \epsilon _V^2}{3}-\frac{4 \epsilon _V \eta _V}{3}
+ \frac{5 \epsilon _V^3}{9}+\frac{16}{3} \epsilon _V^2 \eta _V-\frac{19}{9} \epsilon _V \eta _V^2-\frac{4}{9} \epsilon _V \xi _V^2+\frac{7}{3}
   \epsilon _V^2 \delta _H \right. \\
& \left.   
   -\frac{4}{3} \epsilon _V \eta _V \delta _H-\epsilon _V \delta _H^2
+ \frac{281 \epsilon _V^4}{9}-\frac{1000}{27} \epsilon _V^3 \eta _V+\frac{152}{9} \epsilon _V^2 \eta _V^2-\frac{64}{27} \epsilon _V \eta
   _V^3+\frac{8}{3} \epsilon _V^2 \xi _V^2 \right. \\
& \left.   
 -\frac{16}{9} \epsilon _V \eta _V \xi _V^2-\frac{4}{27} \epsilon _V^2 \sigma _V^2-\frac{4}{27}
   \epsilon _V \sigma _V^3-\frac{49}{6} \epsilon _V^3 \delta _H+\frac{119}{9} \epsilon _V^2 \eta _V \delta _H-\frac{38}{9} \epsilon _V \eta
   _V^2 \delta _H \right. \\
& \left.   
 -\frac{8}{9} \epsilon _V \xi _V^2 \delta _H+\frac{4}{3} \epsilon _V^2 \delta _H^2-2 \epsilon _V \delta _H^3
 \right).
\end{split}
\end{equation}

With use of the obtained results one can also find slow-roll approximations for two further useful quantities :
\begin{equation}
\begin{split}
\frac{a'}{a} = & aH \simeq - \frac{1}{\tau} 
\left( 1 + \epsilon_V 
+ \frac{11 \epsilon _V^2}{3}-\frac{4 \epsilon _V \eta _V}{3}
+ \frac{5 \epsilon _V^3}{9}+\frac{16}{3} \epsilon _V^2 \eta _V-\frac{19}{9} \epsilon _V \eta _V^2-\frac{4}{9} \epsilon _V \xi _V^2+\frac{7}{3}
   \epsilon _V^2 \delta _H \right. \\
& \left.   
   -\frac{4}{3} \epsilon _V \eta _V \delta _H-\epsilon _V \delta _H^2
+ \frac{281 \epsilon _V^4}{9}-\frac{1000}{27} \epsilon _V^3 \eta _V+\frac{152}{9} \epsilon _V^2 \eta _V^2-\frac{64}{27} \epsilon _V \eta
   _V^3+\frac{8}{3} \epsilon _V^2 \xi _V^2 \right. \\
& \left.   
 -\frac{16}{9} \epsilon _V \eta _V \xi _V^2-\frac{4}{27} \epsilon _V^2 \sigma _V^2-\frac{4}{27}
   \epsilon _V \sigma _V^3-\frac{49}{6} \epsilon _V^3 \delta _H+\frac{119}{9} \epsilon _V^2 \eta _V \delta _H-\frac{38}{9} \epsilon _V \eta
   _V^2 \delta _H \right. \\
& \left.   
 -\frac{8}{9} \epsilon _V \xi _V^2 \delta _H+\frac{4}{3} \epsilon _V^2 \delta _H^2-2 \epsilon _V \delta _H^3
 \right),
\end{split}
\end{equation}
\begin{equation}
\begin{split}
\frac{a''}{a} =& a^2 H^2 \left( 1 + \frac{\dot{H}}{H^2} \right) + \left( \frac{a'}{a} \right)  
\simeq \frac{1}{\tau^2} 
\left( 2 + 3 \epsilon_V 
+ 16 \epsilon _V^2-6 \epsilon _V \eta _V 
\frac{23 \epsilon _V^3}{3}+\frac{62}{3} \epsilon _V^2 \eta _V \right. \\
& \left.
 -9 \epsilon _V \eta _V^2-2 \epsilon _V \xi _V^2+11 \epsilon _V^2 \delta _H-6
   \epsilon _V \eta _V \delta _H-3 \epsilon _V \delta _H^2
+ \frac{1478 \epsilon _V^4}{9}-\frac{1588}{9} \epsilon _V^3 \eta _V \right. \\
& \left.
 +\frac{670}{9} \epsilon _V^2 \eta _V^2-10 \epsilon _V \eta _V^3+\frac{100}{9}
   \epsilon _V^2 \xi _V^2-\frac{70}{9} \epsilon _V \eta _V \xi _V^2-\frac{16}{27} \epsilon _V^2 \sigma _V^2-\frac{2}{3} \epsilon _V \sigma
   _V^3-\frac{65}{2} \epsilon _V^3 \delta _H \right. \\
& \left.
 +56 \epsilon _V^2 \eta _V \delta _H-18 \epsilon _V \eta _V^2 \delta _H-4 \epsilon _V \xi _V^2
   \delta _H+6 \epsilon _V^2 \delta _H^2-6 \epsilon _V \delta _H^3
\right) .
\end{split}
\end{equation}


\begin{thebibliography}{99}

\bibitem{Randall:1995dj}
  L.~Randall, M.~Soljacic and A.~H.~Guth,
  Nucl.\ Phys.\ B {\bf 472} (1996) 377
  [hep-ph/9512439].

\bibitem{HenryTye:2006uv}
  S.-H.~Henry Tye,
  Lect.\ Notes Phys.\  {\bf 737} (2008) 949
  [hep-th/0610221].
  
\bibitem{Allahverdi:2006iq}
  R.~Allahverdi, K.~Enqvist, J.~Garcia-Bellido and A.~Mazumdar,
  Phys.\ Rev.\ Lett.\  {\bf 97} (2006) 191304
  doi:10.1103/PhysRevLett.97.191304
  [hep-ph/0605035].
  
\bibitem{Baumann:2014nda}
  D.~Baumann and L.~McAllister,
  arXiv:1404.2601 [hep-th].
 
\bibitem{Mielczarek:2009zw}
  J.~Mielczarek,
  Phys.\ Rev.\ D {\bf 81} (2010) 063503
  [arXiv:0908.4329 [gr-qc]].

\bibitem{Barrau:2013ula}
  A.~Barrau, T.~Cailleteau, J.~Grain and J.~Mielczarek,
  Class.\ Quant.\ Grav.\  {\bf 31} (2014) 053001
  [arXiv:1309.6896 [gr-qc]].
  
\bibitem{Agullo:2013ai}
  I.~Agullo, A.~Ashtekar and W.~Nelson,
  Class.\ Quant.\ Grav.\  {\bf 30} (2013) 085014
  [arXiv:1302.0254 [gr-qc]].

\bibitem{Mielczarek:2015cja}
  J.~Mielczarek,
  arXiv:1503.08794 [gr-qc].
  
\bibitem{Mielczarek:2015xyu}
  J.~Mielczarek,
  arXiv:1512.08997 [gr-qc].
  
\bibitem{Mukohyama:2010xz}
  S.~Mukohyama,
  Class.\ Quant.\ Grav.\  {\bf 27} (2010) 223101
  [arXiv:1007.5199 [hep-th]].
  
\bibitem{Kiefer:2011cc}
  C.~Kiefer and M.~Kraemer,
  Phys.\ Rev.\ Lett.\  {\bf 108} (2012) 021301
  [arXiv:1103.4967 [gr-qc]].
  
\bibitem{Komatsu:2010fb}
  E.~Komatsu {\it et al.} [WMAP Collaboration],
  Astrophys.\ J.\ Suppl.\  {\bf 192} (2011) 18
  [arXiv:1001.4538 [astro-ph.CO]].
  
\bibitem{Ade:2015tva}
  P.~A.~R.~Ade {\it et al.} [BICEP2 and Planck Collaborations],
  Phys.\ Rev.\ Lett.\  {\bf 114} (2015) 101301
  [arXiv:1502.00612 [astro-ph.CO]].
  
\bibitem{Starobinsky:1980te}
  A.~A.~Starobinsky,
  Phys.\ Lett.\ B {\bf 91} (1980) 99.
  
\bibitem{Bojowald:2006da}
  M.~Bojowald,
  Living Rev.\ Rel.\  {\bf 8} (2005) 11
  [gr-qc/0601085].
  
\bibitem{Ashtekar:2011ni}
  A.~Ashtekar and P.~Singh,
  Class.\ Quant.\ Grav.\  {\bf 28} (2011) 213001
  [arXiv:1108.0893 [gr-qc]].
  
\bibitem{Ashtekar:2006rx}
  A.~Ashtekar, T.~Pawlowski and P.~Singh,
  Phys.\ Rev.\ Lett.\  {\bf 96} (2006) 141301
  [gr-qc/0602086].
  
\bibitem{Cailleteau:2011kr}
  T.~Cailleteau, J.~Mielczarek, A.~Barrau and J.~Grain,
  Class.\ Quant.\ Grav.\  {\bf 29} (2012) 095010
  [arXiv:1111.3535 [gr-qc]].
 
\bibitem{Bojowald:2011aa}
  M.~Bojowald and G.~M.~Paily,
  Phys.\ Rev.\ D {\bf 86} (2012) 104018
  [arXiv:1112.1899 [gr-qc]].
  
\bibitem{Bojowald:2015gra}
  M.~Bojowald and J.~Mielczarek,
  JCAP {\bf 1508} (2015) no.08,  052
  [arXiv:1503.09154 [gr-qc]].
  
\bibitem{Mielczarek:2012tn}
  J.~Mielczarek,
  AIP Conf.\ Proc.\  {\bf 1514} (2012) 81
  [arXiv:1212.3527 [gr-qc]].
  
\bibitem{Mielczarek:2014kea}
  J.~Mielczarek, L.~Linsefors and A.~Barrau,
  arXiv:1411.0272 [gr-qc].
  
\bibitem{Artymowski:2008sc}
  M.~Artymowski, Z.~Lalak and L.~Szulc,
  JCAP {\bf 0901} (2009) 004
  [arXiv:0807.0160 [gr-qc]].
  
\bibitem{Singh:2006im}
  P.~Singh, K.~Vandersloot and G.~V.~Vereshchagin,
  Phys.\ Rev.\ D {\bf 74} (2006) 043510
  [gr-qc/0606032].
  
\bibitem{Ashtekar:2009mm}
  A.~Ashtekar and D.~Sloan,
  Phys.\ Lett.\ B {\bf 694} (2011) 108
  [arXiv:0912.4093 [gr-qc]].
  
\bibitem{Mielczarek:2010bh}
  J.~Mielczarek, T.~Cailleteau, J.~Grain and A.~Barrau,
  Phys.\ Rev.\ D {\bf 81} (2010) 104049
  [arXiv:1003.4660 [gr-qc]].
    
\bibitem{Mielczarek:2013xaa}
  J.~Mielczarek,
  JCAP {\bf 1403} (2014) 048
  [arXiv:1311.1344 [gr-qc]].
  
\bibitem{Schwarz:2001vv}
  D.~J.~Schwarz, C.~A.~Terrero-Escalante and A.~A.~Garcia,
  Phys.\ Lett.\ B {\bf 517} (2001) 243
  [astro-ph/0106020].
  
\bibitem{Schwarz:2004tz}
  D.~J.~Schwarz and C.~A.~Terrero-Escalante,
  JCAP {\bf 0408} (2004) 003
  [hep-ph/0403129].
  
\bibitem{Zhu:2015xsa}
  T.~Zhu, A.~Wang, G.~Cleaver, K.~Kirsten, Q.~Sheng and Q.~Wu,
  Astrophys.\ J.\  {\bf 807} (2015) 1,  L17
  [arXiv:1503.06761 [gr-qc]].
  
\bibitem{Zhu:2015owa}
  T.~Zhu, A.~Wang, G.~Cleaver, K.~Kirsten, Q.~Sheng and Q.~Wu,
  JCAP {\bf 1510} (2015) 10,  052
  [arXiv:1508.03239 [gr-qc]].
  
\bibitem{Liddle:1994dx}
  A.~R.~Liddle, P.~Parsons and J.~D.~Barrow,
  Phys.\ Rev.\ D {\bf 50} (1994) 7222
  [astro-ph/9408015].
  
\bibitem{Ashtekar:2006wn}
  A.~Ashtekar, T.~Pawlowski and P.~Singh,
  Phys.\ Rev.\ D {\bf 74} (2006) 084003
  [gr-qc/0607039].
 
\bibitem{Mielczarek:2010bh}
  J.~Mielczarek, T.~Cailleteau, J.~Grain and A.~Barrau,
  Phys.\ Rev.\ D {\bf 81} (2010) 104049
  [arXiv:1003.4660 [gr-qc]].
   			
\end{thebibliography}
\end{document}